\documentclass[12pt]{amsart}
\usepackage{amsmath, amstext, amsbsy, amssymb}

\newtheorem{theorem}{Theorem}[section]
\newtheorem{lemma}[theorem]{Lemma}

\newtheorem{proposition}{Proposition}[section]

\theoremstyle{definition}
\newtheorem{definition}[theorem]{Definition}
\newtheorem{example}[theorem]{Example}

\theoremstyle{remark}

\numberwithin{equation}{section}

\newcommand{\la}{\lambda}

\begin{document}

\title[On classes of local unitary transformations]
{On classes of local unitary transformations}
\author{Naihuan Jing}
\address{Department of Mathematics,
   North Carolina State Univer\-sity,
   Ra\-leigh, NC 27695-8205, USA}
\address{School of Sciences, South China University of Technology,
Guangzhou 510640, China}
\email{jing@math.ncsu.edu}

\begin{abstract}
We give a one-to-one correspondence between classes of density
matrices under local unitary invariance and the double cosets of
unitary groups. We show that the interrelationship among classes of
local unitary equivalent multi-partite mixed states is independent
from the actual values of the eigenvalues and only depends on the
multiplicities of the eigenvalues. The interpretation in terms of
homogeneous spaces of unitary groups is also discussed.

\end{abstract}

\maketitle

\section{Introduction} \label{S:intro}

Theoretic foundation of quantum computation and quantum information
\cite{NC} has been postulated through the language of quantum
mechanics. At the beginning of the development of quantum mechanics
Schr\"oginger pointed out the unitary freedom of the quantum system
\cite{S, U}, which is one of the spotlights in the whole theory and
has been used in many applications in later developments. In
\cite{HJW} the freedom of purification of quantum states was studied
with help of Schmidt decomposition. Furthermore the probability
distributions given by a given density matrix are characterized by
majorization \cite{N}. These results have played an important role
in quantum statistics, quantum computation and quantum information.

Quantum entanglement is one of the key issues in quantum
computation, and local unitary invariance is an importance aspect in
quantum entanglement and their applications. For instance, several
fast algorithms are discovered based on special properties of local
invariance and quantum entanglements. The maximum entanglement can
also be achieved via local unitary action. Local unitary
transformations have also been used for fractional entanglement and
in other properties (cf. \cite{HHH}). In various investigations of
local unitary equivalence lots of efforts have been made in seeking
possible invariants \cite{AFG, TLT, ZSB, ACFW}. There have also been
results on non-local unitary equivalence \cite{DC} in exploring
feasible quantum gates. In \cite{YLF} a geometric class was proposed
to study the problem for two partite quantum states and it has been
shown that two quantum states are equivalent if and only if their
representation classes are the same. In \cite{FJ} we have proposed a
new operational method to study local unitary equivalence of density
matrices and the relationship with separability.

It is clear that in classification of equivalent density matrices
one first needs to see whether two density operators have the same
eigenvalues or not. Therefore one should separate the easy task of
comparing eigenvalues from the problem and focus on other invariants
of unitary operations. In the current paper we apply this strategy
to a larger quantum system consisting of several quantum sub-systems
and study their unitary properties from the viewpoints of both
global and local pictures.

We first look at density matrices with
 specified
spectrum of eigenvalues and study the properties of the quantum
systems using group-theoretic methods. The local unitary properties
are investigated with the help of fixed point subgroups of the
concerned Lie group and we then set up a one-to-one correspondence
between classes of local unitary equivalences and double cosets of
the unitary group $U(n)$ by the additive and multiplicative Young
subgroups. Next we recast the space of the equivalence classes in
terms of induced representations of the Young subgroups and show
their deep connections with homogeneous spaces of unitary groups.

One of our main results shows that local unitary equivalence depends
only on the multiplicities of the eigenvalues and is independent
from the actual values of the eigenvalues. This means that with
given spectrum of the density matrices the local unitary equivalence
is determined by the combinatorics of the quantum states. We hope
that the combinatorics and invariant theory will shed light in
further investigation of local unitary properties and entanglement
of quantum states.

\section{Local unitary equivalence} \label{S:equivalence}

Let $H$ be a Hilbert space affording the quantum system. If the
quantum system is in a number of normalized quantum states
$|\psi_i\rangle$ with probability $p_i$, then the density operator
for the quantum system is the convex combination
$$ \rho=\sum_ip_i|\psi_i\rangle\langle\psi_i|,$$
where $\sum_ip_i=1$. Equivalently density operators are
characterized as non-negative operators on the Hilbert space with
unit trace.

Among all physically possible quantum systems afforded by $H$, one
needs to judge whether two systems are distinct or not. To study the
equivalence of various quantum systems on $H$, one only needs to
consider the unitary linear group $U(n)$ over $\mathbb C$. Suppose
two density matrices $\rho_i$ have the same set of eigenvalues and
their sets of multiplicities for each eigenvalue are also the same.
Let the corresponding (orthonormal) eigenstates be
$|\phi(\la_i^{(j)})\rangle$ and $|\psi(\la_i^{(j)})\rangle$
respectively, where $j=1, \cdots , m_i$, the multiplicity of the
eigenvalue $\la_i$. Then we can write
\begin{align*}
\rho_1&=\sum_{ij}\la_i|\phi(\la_i^{(j)})\rangle\langle
\phi(\la_i^{(j)})|,\\
\rho_2&=\sum_{ij}\la_i|\psi(\la_i^{(j)})\rangle\langle
\psi(\la_i^{(j)})|.
\end{align*}
We define $g$ to be the linear transformation on $H$ sending
$|\phi(\la_i^{(j)})\rangle\mapsto |\psi(\la_i^{(j)})\rangle$, then
$g\in U(n)$ due to orthogonality of the eigenstates. Then one easily
checks that
$$\rho_2=g\rho_1g^{-1}=g\rho_1g^{\dagger}.$$
In other words density matrices with similar spectral decomposition
are unitary equivalent and all physical features of the quantum
system are captured by properties of the unitary group $U(n)$.

The above discussion can also be viewed in terms of group action.
Let $D(n)$ be the set of density matrices on the Hilbert space $H$
of dimension $n$. The evolution of the density operator $\rho$ is
given by $\rho\longrightarrow g\rho g^{\dagger}$, where $g=g(t) \in
U(n)$ and $\dagger$ is transpose and conjugation. Clearly $g\rho
g^{\dagger}$ is positive and $tr(g\rho g^{\dagger})=tr(\rho)=1$. We
will adopt the standard convention in group theory to write
$hgh^{-1}=g^{h}$ for $g, h\in G$.

\begin{lemma} The kernel of the action $U(n)\times D(n)\longrightarrow D(n)$
is $\{e^{i\theta}I|\theta\in \mathbb R\}\leq U(n)$.
\end{lemma}
\begin{proof}
Suppose $g\rho g^{\dagger}=\rho$ for all density matrices $\rho$. If
$\rho$ is any diagonal matrix with distinct eigenvalues, then $g$
must be diagonal by computation. A generalized (or signed)
permutation matrix is one that each row or column has only one and
only non-zero entry and this non-zero entry is $1$ or $-1$. One
notes that any matrix commutes with a generalized permutation matrix
must be diagonal. To show that $g=aI$ with $a=e^{i\theta}$ a complex
number of unit modulus, we can take positive generalized permutation
matrix, then all diagonal entries of $g$ must be equal.
\end{proof}

For $\rho\in D(n)$ we define the {\it invariant subgroup} $U_{\rho}$
by $U_{\rho}=\{g\in U(n)| \, g\rho g^{\dagger}=\rho\}$. As $g\in
U(n)$ is unitary: $gg^{\dagger}=g^{\dagger}g=I$, we have immediately
that $g\rho g^{\dagger}=\rho$ iff $g\rho=\rho g$. The next statement
 is immediate from definition.

\begin{lemma} \label{L:conj}
For any $h\in GL(n)$ one has
$U_{\rho^h}=hU_{\rho}h^{\dagger}=(U_{\rho})^{h}$.
\end{lemma}

Let $\rho$ be a density matrix of size $n\times n$ and thus is
diagonalizable due to hermiticity. According to basic theory in
linear algebra \cite{L} there exists a unitary matrix $g\in U(n)$
such that
$$\rho=g\, diag(a_1I_{\la_1}, \cdots, a_lI_{\la_l})g^{\dagger},$$
where $\la_i$ are non-negative integers such that
$\la_1+\cdots+\la_l=n$. We can further suppose that $\la_i$ are
arranged in a descending order: $\la_1\geq\cdots\geq\la_l>0$, thus
$\la$ is a partition of $n$. Here $l$ is the number of parts of
$\la$.

\begin{definition} We say a density matrix $\rho$ is of type $\la$,
a partition of $n$, if $\rho$ has eigenvalue multiplicities: $\la_1,
\cdots, \la_l$, where $\la_1+\cdots +\la_l=n$.
\end{definition}

Two density matrices are {\it equivalent} if they are related by a
unitary transformation: $\rho_1=g\rho_2 g^{\dagger}$, $g\in U(n)$.
This is in agreement with the equivalence relation given by the
group action of $U(n)$ on $D(n)$. Thus the invariant subgroups of
equivalent density matrices are conjugate in $U(n)$. We will call
the set of equivalent density matrices of type $\la$ the {\it
equivalent class} of type $\la$, and denote the class by $[\la]$.

\begin{proposition} \label{P:conj} If
$\rho=diag(a_1I_{\la_1}, \cdots, a_lI_{\la_l})$ is a density matrix,
then $U_{\rho}=U(\la_1)\times\cdots  \times U(\la_l)$. More
generally if $\rho$ is of type $\la=(\la_1, \cdots, \la_l)$, then
$U_{\rho}=gU(\la_1)\times\cdots \times U(\la_l)g^{\dagger}$ for some
$g\in U(n)$.
\end{proposition}
\begin{proof} If  $\rho=diag(a_1I_{\la_1}, \cdots, a_lI_{\la_l})$,
it follows from direct computation that
$U_{\rho}=U(\la_1)\times\cdots  \times U(\la_l)$. Then for a general
density matrix $\rho$ of type $\la$, we have that
$\rho=gdiag(a_1I_{\la_1}, \cdots, a_lI_{\la_l})g^{\dagger}$ ($g\in
U=U(n)$), then the result follows from Lemma \ref{L:conj}.
\end{proof}

The classes of density matrices can be characterized by the group
action.

\begin{theorem} Let $[\la]$ be the equivalent class of density
matrices of type $\la$, then $[\la]\simeq U/U_{\rho}$, where $\rho$
is some density matrix of type $\la$.
\end{theorem}
\begin{proof} Let $U=U(n)$. The action $U\times [\la]\longrightarrow [\la]$
given by $(g, \rho)\mapsto g\rho g^{\dagger}$ is transitive. As the
class $[\la]=\{\rho^g|g\in U\}$ for a fixed density matrix $\rho$,
we see that
$$ \rho^{g_1}= \rho^{g_2} \Longleftrightarrow g_1g_2^{-1}\in
U_{\rho}, $$ therefore $U/U_{\rho}\simeq [\la(\rho)]$. It follows
from Proposition \ref{P:conj} that
$U/U_{\rho}=U/gU_{\la(\rho)}g^{-1}\simeq U/U_{\la(\rho)}$
\end{proof}

We now consider local unitary equivalence. Let $H_i$ be two Hilbert
spaces with dimensions $n_i$. We say two density operators $\rho_i$
on the space $H=H_1\otimes H_2$ are equivalent under local
transformation iff $\rho_1=(g_1\otimes g_2)\rho_2 (g_1\otimes
g_2)^{\dagger} $ for some $g_i\in U(n_i)\leq End(H_i)$. We recall
that $\rho_i$ are (globally) equivalent if $\rho_1=g\rho_2
g^{\dagger}$ for $g\in U(n)\leq End(H)$, where $n=n_1n_2$. More
generally, for multi-partite cases the local unitary group is
 $U_{\bf n}=U(n_1)\otimes \cdots \otimes U(n_r)$, where
$n_i=dim(H_i)$.

\begin{theorem} Let $H=H_1\otimes \cdots \otimes H_r$ be the global
Hilbert space with $dim(H_i)=n_i$, and $U_{\bf n}$ as above. Then
the type $\la$ equivalent multi-partite mixed states under local
equivalence are in one-to-one correspondent to double cosets of the
unitary group $U(n)$ by the Young subgroups: $U_{\bf n}\backslash
U/U_{\la}$, where $U_{\la}=U(\la_1)\times \cdots \times U(\la_l)$.
Moreover, the representative of the double coset determined by the
density matrix $\rho$ is given by the matrix of the orthonormal
eigenstates of $\rho$.
\end{theorem}
\begin{proof} Globally equivalent density matrices are
determined by their eigenvalue spectrum and multiplicities. The
class of (locally) equivalent density matrices of type $\la$
consists of
$$g\Lambda g^{\dagger}$$
where
$\Lambda=diag(a_1I_{\la_1}, \cdots, a_lI_{\la_l})$ and $g$ is a
unitary matrix consisting of orthonormal eigenstates of $\rho$. Here
$\Lambda$ is fixed. When two density matrices $\rho_i=g_i\Lambda
g_i^{\dagger}$ are local unitary equivalent, then $\rho_1=k\rho_2
k^{-1}$ for some $k\in U_{\bf n}$. It then follows that
$$g_1\Lambda g_1^{\dagger}=kg_2\Lambda g_2^{\dagger}k^{\dagger}
\Longrightarrow g_1^{-1}kg_2=c\in U_{\Lambda},$$ and $U_{\Lambda}=
U_{\la}=U(\la_1)\times \cdots \times U(\la_l)$ by Proposition
\ref{P:conj}. Therefore $U_{\bf n}g_1 U_{\la}=U_{\bf n}g_2 U_{\la}$.

Conversely, suppose two double cosets $U_{\bf n}g_i U_{\la}$ are the
same. Then $g_1=kg_2c$ for some $k\in U_{\bf n}$ and $c\in U_{\la}$.
It follows that $g_1^{-1}kg_2=c^{-1}\in U_{\lambda}$, thus
$g_1\Lambda g_1^{\dagger}=kg_2\Lambda g_2^{\dagger}k^{\dagger}$. It
is clear that $\rho_1=g_1\Lambda g_1^{\dagger}$ gives rise to a
density matrix and $\rho_2=g_2\Lambda g_2^{\dagger}$ gives rise to a
globally equivalent density matrix (as $\rho_1=\rho_2^k$). Their
corresponding classes (under local equivalence) are also equal:
$$[g_1\Lambda g_1^{\dagger}]=[g_2\Lambda g_2^{\dagger}].$$

\end{proof}

\section{Induced representations}

Let $\la$ be a partition of $n$, $n=\la_1+\cdots+\la_l$, and $\bf n$
be a factorization of $n$, $ n=n_1\cdots n_r$, where $n$ is the
dimension of the underlying (global) Hilbert space. The Young
subgroup $U_{\la}$ associated with the partition $\la$ is the direct
product $U(\la_1)\times \cdots \times U(\la_l)$. From our previous
discussion it is clear that one also needs to consider the
multiplicative Young subgroup $U_{\bf n}=U(n_1)\otimes \cdots
\otimes U(n_r)$ associated to the factorization $\bf n$ of $n$. The
additive Young subgroup $U_{\la}=U(\la_1)\times \cdots \times
U(\la_l)$ can be imbedded into $U(n)$ in the canonical manner:
$$(g_1, \cdots, g_l)\in U(\la_i)\mapsto g_1\times\cdots \times g_l
=\begin{pmatrix} g_1 &  & \\
& \ddots &\\
&  & g_l\end{pmatrix} \in U(n),
$$
while the multiplicative Young subgroup $U_{\bf n}=U(n_1)\otimes
\cdots \otimes U(n_r)$ is imbedded into $U(n)$ via tensor product
$$
g_1\otimes\cdots \otimes g_r\in U_{\bf n} \mapsto g_1\otimes\cdots
\otimes g_r \in U(n) .
$$

 As we have seen in Section \ref{S:equivalence}
 the double cosets given by the Young subgroups $U_{\la}$ and
multiplicative Young subgroup $U_{\bf n}$ are in one-to-one
correspondence of quantum multi-partite systems with dimensions
$n_i$. We now reformulate this correspondence in terms of
representations of the unitary group. We can restrict ourself to the
special unitary group without loss of generality. Let $G=SU(n)$ and
$H=SU_{\la}$, and it is clear that $H$ is a closed subgroup of the
compact Lie group $G$.

Let $\mathbb C$ be the trivial representation of the subgroup $H$.
The induced representation $Ind_{H}^{G}\mathbb C$, as a vector
space, is the space of all complex continuous functions on $H$
satisfying the condition:
$$
f(gh)=h^{-1}f(g), \qquad g\in G, h\in H.
$$
The action of $G$ on the induced representation is given by
\begin{equation}
(g\cdot f)(x)=f(g^{-1}x), \qquad g, x\in G.
\end{equation}
The induced representation $Ind_{H}^{G}\mathbb C$ is isomorphic to
$C^0(G/H, \mathbb C)$, the space of continuous functions on the
cosets $G/H$ \cite{BD}. Here the action is the left translation:
$L(g)f(x)=f(g^{-1}x)$. The canonical basis of the representation is
given by characteristic functions $\phi_{gH}$, indexed by right
 cosets $G/H$. Here
\begin{equation}
\phi_{gH}(g'H)=\begin{cases} 1,& g^{-1}g'\in H\\
0, & g^{-1}g'\notin H
\end{cases}.
\end{equation}
The induced representation $Ind_H^G\mathbb C$ is essentially
equivalent to the set of unitary classes of density matrices as seen
from the following observation.

\begin{example} Suppose $n=1\cdot n$ is the factorization of $n$, then local unitary equivalence
coincides with global unitary equivalence. The double cosets $U_{\bf
n}\backslash U/U_{\la}$ reduces to cosets $U/U_{\la}$. The unitary
equivalence classes (with same set of spectrum) are also in
one-to-one correspondence of the set of partitions $\mathcal P(n)$:
$$n=\la_1+\cdots+\la_l.$$
Each equivalence class is represented by a homogeneous space of the
unitary group $U(n)$ determined by the partition $\la$. The
generating function of the number of unitary equivalence classes is
given by \cite{M}
$$
\sum_{n=0}^{\infty}\#(\mbox{equiv
class})q^n=\prod_{n=1}^{\infty}\frac 1{1-q^n}.
$$
\end{example}

\begin{example} Density matrices on any Hilbert space of prime
dimension are always inseparable. Multipartite states should live on
Hilbert spaces of dimension equal to powers of prime numbers.
\end{example}

Local unitary equivalent classes can also be studied from
representation theoretic viewpoints. Let $K=SU_{\bf n}$
corresponding to the factorization $\bf n$, and $H=SU_{\la}\leq
G=SU(n)$ as above. Both $H$ and $K$ are closed subgroups of $G$. We
consider the restriction $Res_KInd_H^G\mathbb C$. Let $W$ be any
$H$-module, then it follows from Frobenius reciprocity \cite{BD}
that
\begin{align*}
Hom_K( Res_KInd_H^G\mathbb C, W)&\simeq Hom_G(Ind_H^G \mathbb C,
Ind_K^G W)\\
&\simeq Hom_H(\mathbb C, Res_HInd_K^G W).
\end{align*}
In particular, when $W=\mathbb C$, one has the duality
$Res_KInd_H^G\mathbb C\simeq Res_HInd_K^G\mathbb C$. The linear
operators in $Hom_G(Ind_H^G \mathbb C, Ind_K^G \mathbb C)$ in
general are certain distributions on the direct product of the
unitary group by Schwartz's distribution theory \cite{Sc}. If the
intertwining number (the dimension of $Hom_G(Ind_H^G \mathbb C,
Ind_K^G W)$) is finite, then they are equal to the number of double
cosets of $K\backslash G/H$ by Mackey's theorem \cite{Mc}.
 Thus we obtain the following result.

\begin{theorem} Let $H=H_1\otimes \cdots \otimes H_r$ be the global
Hilbert space with $dim(H_i)=n_i$, and $U_{\bf n}$ as above. Then
the classes of globally equivalent multi-partite mixed states under
local equivalence are isomorphic to the restriction of the induced
representation $Ind_{U_{\la}}^{U(n)}\mathbb C$ to the subgroup
$SU_{\bf n}$.
\end{theorem}

\section{Homogeneous spaces}

Homogeneous spaces are important non-empty manifolds with a
transitive action of a Lie group. The quotient space $G/H$ are
special examples of homogeneous spaces. If one views equivalent
local unitary classes as points in the homogeneous spaces one may
understand the geometric meaning of the statement that all points
are the same.

Let $\bf n$ be a fixed prime decomposition of $n$ and $\la$ be a
fixed partition of $n$. If two density matrices $\rho_i$ are locally
equivalent, then their corresponding double cosets $U_{\bf
n}g_1U_{\la}$ and $U_{\bf n}g_1U_{\la}$ are the same. Note that in
this description the equivalence relations are determined completely
by the double cosets, and no information about the actual
eigenvalues are needed. In other words, classes of density matrices
with the same eigenvalue distribution can be mapped to classes of
density matrices with the same eigenvalue multiplicities. Therefore
we have proved the following result.

\begin{theorem} Let $H=H_1\otimes \cdots \otimes H_r$ be the global
Hilbert space with $dim(H_i)=n_i$, and $U_{\bf n}$ as above. Then
the class of globally equivalent multi-partite mixed states under
local equivalence can be mapped to another class of globally
equivalent multi-partite mixed states with the same eigenvalue
distribution and multiplicities. The two classes are both
described
 by the homogeneous spaces $U_{\bf n}\backslash U/U_{\la}$, where
 the partition $\la$ is given by the eigenvalue multiplicities.
\end{theorem}

This result shows that the local unitary equivalence relations does
not depend on the actual values of the eigenvalues, and all classes
can be viewed as ``equal'' in geometric sense. Therefore we can
parameterize the classes of density matrices by the set
$$\mathbb
R^n\times (U_{\bf n}\backslash U(n)/U_{\la}), $$
where $\mathbb R$ is
used for the eigenvalue.

\begin{example}
For $e\geq 0$ and $0\leq f\leq 1-e$ we consider the two qubit Werner
state
$$\rho=\left(\begin{array}{llll}
                      \frac{1-e-f}3 & & & \\
                     & \frac{1+2f}6 & \frac{1-4f}6 &  \\
                     & \frac{1-4f}6 & \frac{1+2f}6 &  \\
                  & & &   \frac{1+e-f}3  \\
\end{array}
\right).  $$ When $e=0$, this is the usual Werner state \cite{W}.
The eigenvalues are
$$
(1-f+e)/3, (1-f)/3, (1-f-e)/3, f.
$$
The nontrivial factorization is certainly $4=2\cdot 2$ and the
partition of the eigenvalue multiplicities are $\la=(1111)$ for
$e>0$ and $f\neq \frac{1\pm e}4$. When $e>0$ and $f=\frac{1\pm e}4$
or $1/4$ then $\la=(211)$. If $e=0$, then $\la=(31)$ for $f\neq 1/4$
and $\la=(4)$ for $f=1/4$. From the picture there are essentially
one dense class of local unitary equivalence; two one-dimensional
classes of unitary equivalence; and one degenerate class of unitary
equivalence. The most interesting cases are $\la=(31)$ and $(21^2)$,
where the local unitary classes are classified by double cosets of
$SU_{(22)}\backslash SU(4)/(SU(3)\times SU(1))$ and
$SU_{(22)}\backslash SU(4)/(SU(2)\times SU(1)\times SU(1))$
respectively. All classes with the same partition are viewed as
equal.
\end{example}

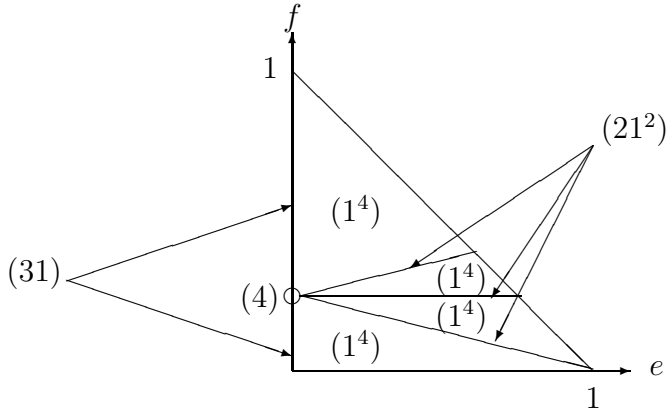
\setlength{\unitlength}{1cm}
\begin{picture}(5.5,5.5)(-4.5,-0.25)
\put(0,0){\vector(1,0){4.5}} \put(4.75,-.05){$e$}
\put(0,0){\vector(0,1){4.5}} \put(0,4.7){\makebox(0,0){$f$}}
\put(0,4){\line(1,-1){4}}
\put(0,1){\circle{0.2}}\multiput(0.2,1)(0.3,0){10}{\line(1,0){0.15}}
\put(0.1,1){\line(4,1){2.35}} \put(0.1,1){\line(4,-1){3.9}}
\put(0.1,1){\line(1,0){2.9}}
\put(-0.7,0.85){$(4)$}\put(3.9,-0.5){$1$} \put(-0.4,3.9){$1$}

\put(1.9,1.1){$(1^4)$}\put(1.9,0.6){$(1^4)$}\put(0.5,0.2){$(1^4)$}
\put(0.5,2){$(1^4)$}

\put(4,3){\vector(-2,-3){1.35}} \put(4,3){\vector(-3,-2){2.44}}
\put(4,3){\vector(-1,-2){1.3}}\put(4.1,3.1){$(21^2)$}

\put(-3,1.2){\vector(3,1){3}}
\put(-3,1.2){\vector(3,-1){3}}\put(-3.8,1.2){$(31)$}
\end{picture}

\section{Conclusion}

We have established a one-to-one correspondence between the set of
density matrices of multi-partite systems and the set of double
cosets of additive and multiplicative Young subgroups of unitary
groups. Our results show that the interrelationship among the
classes of local unitary equivalent multi-partite mixed states is
independent from the actual values of the eigenvalues and only
depends on the multiplicities of the eigenvalues. This interesting
phenomenon allows us to look at the local unitary equivalence for
all global unitary density matrices at the same time. The geometry
of the classes are given by that of homogeneous spaces of the
unitary group as well as the invariant theory of classical groups
(cf. \cite{H}). It is expected that the combinatorics of density
matrices and invariant theory will also play a role in this study.

\bigskip
\centerline{ Acknowledgments}

We would like to thank S. M. Fei for stimulating discussion at the
early stage of the work. We are also grateful to the support of NSA
grants and NSFC's Distinguished Youth Grant. Parts of the work was
conceived while the author was visiting the Chern Institute of
Mathematics in Tianjin, China.

\bibliographystyle{amsalpha}

\end{document}